\documentclass[noshowkeys,amsmath,pra,twocolumn]{revtex4}

\usepackage{times} 
\usepackage{amsmath,amssymb,amsthm} 
\usepackage{bbm} 
\usepackage{dsfont} 
\usepackage[colorlinks=true, pdfstartview=FitH, linkcolor=black,citecolor=black, urlcolor=black,]{hyperref} 
\usepackage[latin1]{inputenc} 
\usepackage[usenames]{color} 
\usepackage{graphicx}

\newcommand{\Cx}{{\mathbb C}}

\newcommand{\idty}{\mathds{1}}

\DeclareMathOperator{\tr}{Tr}

\providecommand{\norm}[1]{\lVert#1\rVert}

\providecommand{\ket}[1]{|#1\rangle}
\providecommand{\bra}[1]{\langle#1|}

\renewcommand{\c}[1]{\mathcal{#1}}
\newcommand{\g}[1]{\mathfrak{#1}}

\renewcommand{\r}[1]{\mathrm{#1}}

\newtheorem{theorem}{Theorem}

\theoremstyle{remark}

\hypersetup{
    pdftitle={Additivity of the Renyi entropy of order 2 for Positive Partial Transpose inducing channels},
    pdfauthor={B.~Dierckx, M.~Fannes and C.~Vandenplas},
} 

\begin{document}

\title{Additivity of the Renyi entropy of order 2 for positive-partial-transpose-inducing channels}
\author{B.~Dierckx}
\email{brecht.dierckx@fys.kuleuven.be}
\author{M.~Fannes}
\email{mark.fannes@fys.kuleuven.be}
\author{C.~Vandenplas}
\email{caroline@itf.fys.kuleuven.be}
\affiliation{Institute for Theoretical Physics, KULeuven, Celestijnenlaan 200D, B-3001 Heverlee, Belgium}
\date{\today}


\begin{abstract}
We prove that the minimal Renyi entropy of order 2 (RE2) output of a positive-partial-transpose(PPT)-inducing channel joint to an arbitrary other channel is equal to the sum of the minimal RE2 output of the individual channels. PPT-inducing channels are channels with a Choi matrix which is bound entangled or separable.
The techniques used can be easily recycled to prove additivity for some non-PPT-inducing channels such as the depolarizing and transpose depolarizing channels, though not all known additive channels.  We explicitly make the calculations for  generalized Werner-Holevo channels as an example of both the scope and limitations of our techniques.
\end{abstract}

\maketitle



Some of the most profound questions in quantum information theory are additivity and multiplicativity questions \cite{Ruskai2007}, where one has to ascertain whether some quantity is more than the sum (or product) of its parts or not.  Though there is a plethora of such questions around, Schor \cite{Shor2003} has proven that the physically relevant ones are almost all equivalent  and that their meaning can be reduced to a single physical question: ``Does multiplexing of quantum channels augment their usefulness for communication by using entangled input states?''  Or, stated more precisely: ``Is the minimal von Neumann entropy output of two or more channels additive?''
\begin{equation*}
\min_{\rho} \r{S} \bigl( \Lambda\otimes \Gamma (\rho) \bigr) \stackrel{?}{=} \min_{\sigma_1} \r{S}\bigl(\Lambda(\sigma_1)\bigr) + \min_{\sigma_2} \r{S}\bigl(\Gamma(\sigma_2)\bigr)
\end{equation*}

Incidently, this question is also thought to be the easiest additivity conjecture to prove and for a long time the community was confident that the way to prove this result, was to prove another additivity conjecture concerning Renyi entropies.
\begin{equation*}
\label{eq:renyi}
H_p (\rho) := \frac{1}{1-p} \log{\tr \rho^p}
\end{equation*}
Namely that for $1 < p \leq 2$ the minimal $p$-Renyi entropy output of two channels is additive.  By taking the limit for $p \rightarrow 1$, this would then prove the additivity conjecture for minimal von Neumann entropy output.  For close to ten years people have proven additivity results of this type for subclasses of quantum channels and almost all constructive examples seemed to obey the additivity conjecture for Renyi entropies with $p$'s in the relevant interval.

Recently, however, there has been a series of articles \cite{Hayden2007,Cubitt2007,Winter2007} that provide counterexamples for all values of $p$ different from one  and so this approach to the question of additivity of minimal von Neumann entropy output fails.  Because of this, the interest in Renyi entropies waned some, but to the authors it seems that the question when these conjectures fail exactly, is a pertinent one.  The scheme by which counterexamples were constructed in \cite{Hayden2007,Cubitt2007,Winter2007} fails for $p=1$ and there is no clear indication what is so intrinsically different for this value.  A possible way to answer this question is to try and determine which channels actually do obey the additivity conjecture for values of $p$ different from 1 and look for a physical or at least mathematically manageable criterium to separate them from the ones that violate the conjecture.  This type of approach led Audenaert \cite{Audenaert2007} to propose another attack vector to the problem of additivity by considering the dimension of the input space of the channels as a parameter of the problem.  He suggests that there might be an appropriate scaling regime dependent on the value of $p$ and the dimension of the input space where the additivity conjecture for minimal Renyi entropy output holds.  In this way one could avoid the counterexamples \cite{Hayden2007,Cubitt2007,Winter2007} to the original conjecture and still take the limit $p \rightarrow 1$.

In this article we take another look at the additivity question for $p=2$.  Our main result, \autoref{thm:main}, proves that for PPT-inducing channels \footnote{A map $\Lambda$ is called PPT-inducing if it is completely positive and $\Lambda \circ T$ is also completely positive.} the additivity conjecture \emph{is} valid and furthermore that for joint use of a PPT-inducing channel and any other (non-PPT-inducing) channel additivity also holds. We singled out the value $p=2$ since this allowed us to use the physical picture of mean field Hamiltonians \cite{Zanardi2004,Fannes2006} and also because, for $p=2$, the Renyi entropy is closely related to the concept of  purity of a channel and, in fact, our results imply multiplicativity for the maximal purity of PPT-inducing channels.

The explicit condition required to make our proof work is actually more powerful than what we need just for additivity of PPT-inducing channels.  Unfortunately, this condition does not determine a convex set with nice properties.  The most general form can be constructed as the union of that set and a deformation of it (\autoref{thm:positivekl}) but again this does not correspond to any known class of channels.  In \autoref{sec:discussion} we give an instructive example to show how in general our criteria work.


\section{A replica trick}
\label{sec:replica}

For the calculation of the 2-norm of a state, or more to the point, of the output state of a quantum channel $\Lambda$, we can use a replica trick to linearize the problem:
\begin{align}
\norm{\Lambda(\rho)}_2^2    &= \tr \Lambda(\rho)^2 \notag \\
    &= \tr \Bigl( \sum_{\alpha } v_{\alpha} \rho v_{\alpha}^* \Bigr) \Bigl( \sum_{\beta} v_{\beta} \rho v_{\beta}^* \Bigr) \notag \\
    &= \sum_{i, j , \alpha , \beta} \bra{e_i} \rho\, v_{\alpha}^* v_{\beta} \ket{e_j} \bra{e_j} \rho\, v_{\beta}^* v_{\alpha} \ket{e_i} \label{eq:replica1} \\
    &= \sum_{i, j , \alpha , \beta} \bra{e_i \otimes e_j} \bigl(\rho \, v_{\alpha}^* v_{\beta} \otimes \rho \, v_{\beta}^* v_{\alpha}\bigr) \, F \ket{e_i \otimes e_j} \notag \\
    &= \tr \rho \otimes \rho \,  \Bigl(\sum_{\alpha , \beta} v_{\alpha}^* v_{\beta} \otimes v_{\beta}^* v_{\alpha}\, \Bigr) F \notag
\end{align}
where $F$ is the `flip' operation on the Hilbert space $\c H \otimes \c H$
\begin{equation*}
F \ket{\varphi \otimes \psi} = \ket{\psi \otimes \varphi}.
\end{equation*}
So, to any quantum channel $\Lambda$ working on $\g B (\c H)$ defined by a set of Kraus operators $\{ v_{\alpha}\}$ as
\begin{equation*}
\Lambda(\rho) = \sum_{\alpha} v_{\alpha} \rho  v_{\alpha}^* \, , \quad \sum_{\alpha} v_{\alpha}^* v_{\alpha} = \idty ,
\end{equation*}
we associate the two particle interaction
\begin{equation}
\label{eq:hamlambda}
{h_{\Lambda}}_{i,j} = - \Bigl(\sum_{\alpha, \beta} v_{\alpha} v_{\beta}^* \otimes v_{\beta} v_{\alpha}^*\Bigr) \, F \, .
\end{equation}

We can treat the indices $i$ and $j$ as belonging to a countable index set $Z$. From the ${h_{\Lambda}}_{i,j}$ we construct a mean field Hamiltonian $H$
\begin{equation}
\label{eq:meanfield}
H = \lim_{|Z| \rightarrow \infty} \frac{1}{|Z|} \sum_{i , j \in Z} {h_{\Lambda}}_{i,j} \, .
\end{equation}

Because of the permutation invariance of this type of Hamiltonians and de Finetti's theorem \cite{Finetti1937}, the thermal and ground states must be exchangeable \footnote{Exchangeable states are mixtures of states of the form $\rho \otimes \rho \otimes \cdots$} states, which in turn implies that the ground state energy density of the Hamiltonian model \eqref{eq:meanfield} is equal to minus the maximal 2-norm of the channel $\Lambda$.  Furthermore, because of the convexity of the 2-norm, the maximum is achieved on a pure state.  So the set over which we optimize can be reduced to the bosonic exchangeable states, i.e., states of the form
\begin{equation*}
\sigma = \ket{\varphi}\bra{\varphi} \otimes \ket{\varphi}\bra{\varphi} \otimes \cdots .
\end{equation*}
In terms of quantum channels, this means that the problem of calculating the maximal 2-norm output, or purity, of a channel is equivalent to finding the maximum of
\begin{align*}
 \bra{\varphi \otimes \varphi}\Bigl( \sum_{\alpha, \beta}  v_{\alpha} v_{\beta}^* \otimes v_{\beta} v_{\alpha}^*\Bigr)\, F \, \ket{\varphi \otimes \varphi} .
\end{align*}

For a tensor product of two channels, $\Lambda$ working on $\g B (\c H)$ and $\Gamma$ working on $\g B (\c K)$, a calculation similar to \eqref{eq:replica1} gives us the expression for $h_{\Lambda \otimes \Gamma}$
\begin{align*}
h_{\Lambda \otimes \Gamma} &= \Bigl(\sum_{\alpha, \beta, \gamma, \delta} v_{\alpha}^* v_{\beta} \otimes w_{\gamma}^* w_{\delta} \otimes v_{\beta}^* v_{\alpha} \otimes w_{\delta}^* w_{\gamma}\Bigr) \, F_{1,3} \, F_{2,4}\\
    &=F_{2,3} \, \bigl( h_{\Lambda} \otimes h_{\Gamma} \bigr) \, F_{2,3} \, ,
\end{align*}
and a similar optimization problem
\begin{equation}
\label{eq:optprob2}
\max_{\varphi \, \in \c H \otimes \c K} \,  \bra{\varphi \otimes \varphi}\,  F_{2,3} \bigl( h_{\Lambda} \otimes h_{\Gamma} \bigr) F_{2,3}\, \ket{\varphi \otimes \varphi} \, .
\end{equation}
In general, the presence of the `flips' $F_{2,3}$ in the expression \eqref{eq:optprob2} prevents us from directly proving additivity results.  However, in \cite{Fannes2006} similar Hamiltonians were treated and (sufficient) conditions identified under which the ground state energy of products of such Hamiltonians is multiplicative.  By a similar analysis we can now compute the corresponding sufficient conditions under which the maximal 2-norm for joint use of channels is additive.

First we rewrite the expression \eqref{eq:hamlambda} in a form which relates more directly to the action of the corresponding quantum channel:
\begin{align}
-h_{\Lambda} &=  \Bigl(\sum_{\alpha, \beta} v_{\alpha} v_{\beta}^* \otimes v_{\beta} v_{\alpha}^*\Bigr) \, F \notag\\
&= \sum_{\alpha, \beta}\Bigl( v_{\alpha} \otimes v_{\beta}\Bigr) \, F \, \Bigl(v_{\alpha}^* \otimes v_{\beta}^*\Bigr) \notag \\
&= \Lambda^* \otimes \Lambda^* (F ) \notag \\
&= \Lambda^* \otimes \Lambda^* \circ \bigl( \textrm{id} \otimes T\bigr) \bigl( \ket{\psi^+} \bra{\psi^+} \bigr) \notag\\
&= \Lambda^* \otimes (\Lambda^* T) \bigl(\ket{\psi^+} \bra{\psi^+}\bigr) \label{eq:positive}
\end{align}
where we have introduced the notation $\Lambda^*$ for the dual map, $T$ for the transposition and $\ket{\psi^+} \bra{\psi^+}$ for the unnormalized maximally entangled state $\sum_{i,j} \ket{e_i \otimes e_i}\bra{e_j \otimes e_j}$.

 Note that, although it is not a necessary condition, this implies that, at least when $\Lambda$ is a PPT-inducing channel \footnotemark[11], the operator $-h_{\Lambda}$ is positive.


\section{Additivity of the 2-Renyi entropy}

We have now gathered enough notation and observations so that we can state and prove the main result of this paper.

\begin{theorem}
\label{thm:main}
For any PPT-inducing channel $\Lambda$ and any other channel $\Gamma$, the minimal 2-Renyi entropy output after a joint use of the channels is equal to the sum of the minimal entropy output of their separate uses,
\begin{equation*}
\min_{\rho} H_2 \Bigl( \Lambda \otimes \Gamma (\rho) \Bigr) = \min_{\sigma_1} H_2 \Bigl(\Lambda(\sigma_1) \Bigr) + \min_{\sigma_2} H_2 \Bigl( \Gamma (\sigma_2) \Bigr) \, .
\end{equation*}
\end{theorem}

\begin{proof}
Consider the functional
\begin{equation}
\label{eq:unnormedstate}
\g B(\c K \otimes \c K) \rightarrow \Cx \, : \, X \mapsto \bra{\varphi \otimes \varphi} F_{2,3}\, \bigr( h_{\Lambda} \otimes X \bigl) \, F_{2,3} \ket{\varphi \otimes \varphi}\, .
\end{equation}
This is clearly a positive functional if $h_{\Lambda}$ is positive, which is true whenever $h_{\Lambda}$ comes from a PPT-inducing channel.  We can normalize \eqref{eq:unnormedstate} to get a state $\omega_{\varphi}$:
\begin{equation}
\label{eq:normedstate}
\omega_{\varphi} \, = \, X \mapsto \frac{\bra{\varphi \otimes \varphi} F_{2,3}\, \bigl( h_{\Lambda} \otimes X \bigr) \,F_{2,3} \ket{\varphi \otimes \varphi}}{\bra{\varphi \otimes \varphi} F_{2,3} h_{\Lambda} \otimes \idty F_{2,3} \ket{\varphi \otimes \varphi}} \, .
\end{equation}
Note also that, for any $Y \in \g B (\c K)$,
\begin{align*}
\label{eq:exchangepos}
\omega_{\varphi} ( Y \otimes Y^* ) &=  \frac{ \bra{\varphi \otimes \varphi} F_{2,3} \bigl(h_{\Lambda} \otimes Y \otimes Y^* \bigr) F_{2,3}\ket{\varphi \otimes \varphi}}{\bra{\varphi \otimes \varphi} F_{2,3} h_{\Lambda} \otimes \idty F_{2,3} \ket{\varphi \otimes \varphi}} \\
&= \sum_{\alpha , \beta} \frac{\bra{\varphi}\,  v_{\alpha}^* v_{\beta} \otimes Y \ket{\varphi} \bra{\varphi} v_{\beta}^* v_{\alpha} \otimes Y^* \, \ket{\varphi}  }{\bra{\varphi \otimes \varphi} F_{2,3} h_{\Lambda} \otimes \idty F_{2,3} \ket{\varphi \otimes \varphi}} \geq 0 \, .
\end{align*}
According to Theorem 2 from \cite{Fannes2006} this means that \eqref{eq:normedstate} is an exchangeable state, i.e., the density matrix $\rho_{\omega}$ is of the form
\begin{equation}
\label{eq:productstate}
\rho_{\omega} = \sum_i \lambda_i \, \sigma_i \otimes \sigma_i \, , \quad \sum_i \lambda_i = 1.
\end{equation}

So, first we rewrite the optimization problem \eqref{eq:optprob2}
\begin{equation*}
\max_{\varphi} \bra{\varphi \otimes \varphi}\, F_{2,3} \bigl( h_{\Lambda} \otimes h_{\Gamma}\bigr) F_{2,3} \, \ket{\varphi \otimes \varphi}
\end{equation*}
as
\begin{equation*}
\max_{\varphi} \omega_{\varphi}\bigl(h_{\Gamma} \bigr) \bra{\varphi \otimes \varphi} F_{2,3} \bigl(h_{\Lambda} \otimes \idty \bigr)F_{2,3} \ket{\varphi \otimes \varphi} \, .
\end{equation*}
This is surely smaller than the product of the maximum of both factors,
\begin{equation*}
\leq \max_{\psi} \omega_{\psi} \bigl(h_{\Gamma} \bigr) \max_{\varphi} \bra{\varphi \otimes \varphi} F_{2,3} \bigl( h_{\Lambda} \otimes \idty \bigr) F_{2,3} \ket{\varphi \otimes \varphi} \, ,
\end{equation*}
and, because of the product nature of the state $\omega_{\psi}$ implied by \eqref{eq:productstate}, the first maximum is smaller than the maximal 2-norm of $h_{\Gamma}$.  The second maximum is exactly the maximal 2-norm of $h_{\Lambda}$.  So,
\begin{equation*}
\max_{\rho} \norm{\Lambda \otimes \Gamma (\rho)}_2^2 \leq \max_{\sigma_1} \norm{\Lambda(\sigma_1)}_2^2 \, . \, \max_{\sigma_2} \norm{\Gamma(\sigma_2)}_2^2 \, .
\end{equation*}
The inequality in the reverse direction,
\begin{equation*}
\max_{\rho} \norm{\Lambda \otimes \Gamma (\rho)}_2^2 \geq \max_{\sigma_1} \norm{\Lambda(\sigma_1)}_2^2 \, . \, \max_{\sigma_2} \norm{\Gamma(\sigma_2)}_2^2
\end{equation*}
we get for free, because the set over which is optimized in the RHS, is included in the set over which it optimized in the LHS.
\\
\end{proof}


\section{Discussion \label{sec:discussion} }

In our proof for \autoref{thm:main} we could have replaced the positivity condition on $h_{\Lambda}$ by a positivity condition on $h_{\Lambda} F$.  The rest of the proof would then be analogous to the original one.  Unfortunately, neither the set for which $h_{\Lambda}$ is positive, nor the set for which $h_{\Lambda} F$ is positive, is convex.  As a whole, this makes the channels which satisfy the consequent conditions, rather awkward to handle.  However, the conditions are computationally rather easy to check and may serve as a quick and easy proof of additivity for a specific channel under study.
\begin{theorem} \label{thm:positivekl}
For any two channels $\Lambda$ and $\Gamma$ such that at least one of the following inequalities is satisfied
\begin{equation}
\begin{cases}
&\Lambda^* \otimes (\Lambda^*T) \,\bigl(\ket{\psi^+}\bra{\psi^+}\bigr) \geq 0, \\
&\Lambda^* \otimes (\Lambda^*T) \,\bigl(\ket{\psi^+}\bra{\psi^+}\bigr)\, .\,F \geq 0, \\
&\Gamma^* \otimes (\Gamma^*T) \,\bigl(\ket{\psi^+}\bra{\psi^+}\bigr) \geq 0, \\
&\Gamma^* \otimes (\Gamma^*T) \,\bigl(\ket{\psi^+}\bra{\psi^+}\bigr)\, .\,F \geq 0, \\
\end{cases}
\end{equation}
the minimal RE2 output after a joint use of the channels is equal to the sum of the minimal entropy output of their separate uses,
\begin{equation*}
\min_{\rho} H_2 \Bigl( \Lambda \otimes \Gamma (\rho) \Bigr) = \min_{\sigma_1} H_2 \Bigl(\Lambda(\sigma_1) \Bigr) + \min_{\sigma_2} H_2 \Bigl( \Gamma (\sigma_2) \Bigr) \, .
\end{equation*}
\end{theorem}

As an example, we constructed a geometric picture of all the above sets for the subclass of generalized Werner-Holevo channels (in dimension $d$):
\begin{equation}
\Lambda (\rho) = a \rho + b \rho^T + (1-a-b) (\tr \rho) \frac{\idty}{d} .
\end{equation}
Such an operation is completely positive iff
\begin{equation}
\label{eq:WHCP}
\begin{cases}
&a (d^2-1) + b (d-1) + 1 \geq 0, \\
&b (d-1) - a +1 \geq 0, \\
&b (d+1) + a - 1 \leq 0. \\
\end{cases}
\end{equation}
The minimal RE2 output can be easily computed.  The set of generalized WH channels splits up into two parts depending on the sign of $ab$
\begin{equation}
\begin{cases}
\min_{\rho} H_2 \Bigr(\Lambda(\rho)\Bigl) = \frac{1+(d-1)(a+b)^2}{d} & ab \geq 0, \\
\min_{\rho} H_2 \Bigr(\Lambda(\rho)\Bigl) = \frac{d(a^2+b^2) - (a+b)^2 + 1}{d} & ab \leq 0.
\end{cases}
\end{equation}

If, in addition to \eqref{eq:WHCP}, the following equations are satisfied
\begin{equation}
\label{eq:WHPPT}
\begin{cases}
&b (d^2-1) + a (d-1) + 1 \geq 0, \\
&a (d-1) - b +1 \geq 0, \\
&a (d+1) + b - 1 \leq 0. \\
\end{cases}
\end{equation}
the WH channel is also PPT-inducing.

The sets described by these equations are graphically depicted in \autoref{fig:WH} for $d=10$.  The image can be considered to be a generic case.  All qualitative features remain the same regardless of the dimension.
\begin{figure}[h]
  \includegraphics[width=8.6cm]{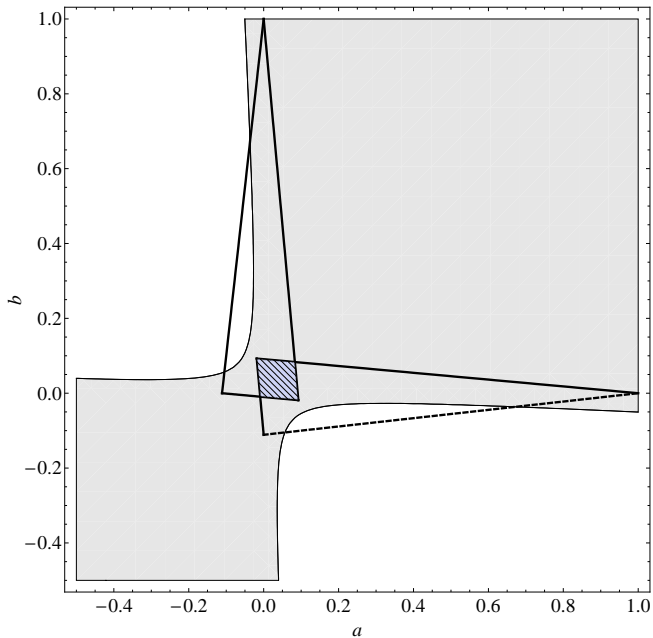}\\
  \caption{A graphic depiction of the (PPT-inducing) generalized Werner-Holevo channels in terms of the parameters $a$ and $b$.  The hatched quadrangle is the set of PPT-inducing gen.\ WH channels.  The gray area is where the conditions of \autoref{thm:positivekl} are satisfied.  The two triangles are the gen.\ WH channels and their image under transposition. The depolarized WH channels correspond to the bottom line of the lower triangle.}\label{fig:WH}
\end{figure}
 Note that the positivity conditions from \autoref{thm:positivekl} are not met for a significant subsection of the gen.\ WH channels, though for the depolarizing WH channels  Michalakis \cite{Michalakis2007} recently proved an additivity result for the special case of joint use of two identical depolarized WH channels .

The set of PPT-inducing generalized Werner-Holevo channels is a polytope with four extremal points.  As a last remark we would like to mention that this quadrangle contains channels whose Choi matrix is separable as well as channels whose Choi matrix is bound entangled.  As an example, we mention the upper right extremal point of the simplex with parameters
\begin{equation}
\begin{cases}
& a = \frac{-2}{-2 +d(d+1)},\\
& b = \frac{d}{-2 + d + d^2}.
\end{cases}
\end{equation}
Its Choi matrix is
\begin{equation}
\frac{2 d}{d^2 + d -2} \Bigl( \frac{\idty + F}{2} - \ket{\psi^+} \bra{\psi^+} \Bigr).
\end{equation}
The range criterium then shows that this Choi matrix is not separable.  It is, however, PPT and thus bound entangled.


\begin{acknowledgments}
We would like to thank M.B. Ruskai for useful comments and suggestions regarding the examples of PPT-inducing channels.

Partial funding was provided by the Belgian Interuniversity Attraction Poles Programme P6/02.
\end{acknowledgments}

\end{document}